\documentclass[conference]{IEEEtran}%
\usepackage{cite}
\usepackage{multicol}
\usepackage{amsmath}
\usepackage{subeqnarray}
\usepackage{cases}
\usepackage{framed}
\usepackage{amsfonts}
\usepackage{amssymb}
\usepackage{bm}
\usepackage{graphicx}
\usepackage{subfig}
\usepackage[english]{babel}
\usepackage{color}%
\usepackage{algorithm}
\usepackage{algpseudocode}
\usepackage{url}
\usepackage[belowskip=-6pt,aboveskip=0pt]{caption}
\usepackage{amsthm}
\usepackage{amsbsy}
\newtheorem{thm}{Theorem}

\theoremstyle{remark}

\theoremstyle{definition}

\providecommand{\U}[1]{\protect\rule{.1in}{.1in}}

\newcommand{\V}{\mathcal{V}}
\newcommand{\xb}{\mathbf{x}}

\newcommand{\Rt}{\mathbb{R}^2}

\newcommand{\Exp}{\text{Exp}}
\newcommand{\Ebb}{\mathbb{E}}
\newcommand{\lt}{\mathcal{L}}

\begin{document}

\title{\LARGE Full Duplex Emulation via Spatial Separation of Half Duplex Nodes in a Planar Cellular Network}

\author{\IEEEauthorblockN{Henning Thomsen, Dong Min Kim, Petar Popovski, Nuno K. Pratas, Elisabeth de Carvalho}
\IEEEauthorblockA{Department of Electronic Systems, Aalborg University, Denmark\\
Email: \{ht, dmk, petarp, nup, edc\}@es.aau.dk} }

\maketitle


\begin{abstract}
A Full Duplex Base Station (FD-BS) can be used to serve simultaneously two Half-Duplex (HD) Mobile Stations (MSs), one working in the uplink and one in the downlink, respectively. The same functionality can be realized by having two interconnected and spatially separated Half Duplex Base Stations (HD-BSs), which is a scheme termed \emph{CoMPflex} (CoMP for In-Band Wireless Full Duplex). A FD-BS can be seen as a special case of CoMPflex with separation distance zero. In this paper we study the performance of CoMPflex in a two-dimensional cellular scenario using stochastic geometry and compare it to the one achieved by FD-BSs. By deriving the Cumulative Distribution Functions, we show that CoMPflex brings BSs closer to the MSs they are serving, while increasing the distance between a MS and interfering MSs. Furthermore, the results show that CoMPflex brings benefits over FD-BS in terms of communication reliability. Following the trend of wireless network densification, CoMPflex can be regarded as a method with a great potential to
effectively use the dense HD deployments.
\end{abstract}

\begin{IEEEkeywords}
Full Duplex, Cellular Communications, Clustering, Network Densification, Spatial Model.
\end{IEEEkeywords}


\section{Introduction}\label{sec:Introduction}

As the wireless cellular networks evolve towards the 5G generation, it is
expected that the number of Base Stations (BSs) per area will noticeably
increase~\cite{boccardi2014vodafone}, leading to \emph{network
densification}. The availability of multiple proximate and interconnected BSs leads to
the usage of cooperative transmission/reception techniques, commonly referred to as
\emph{Coordinated Multi-Point (CoMP)}~\cite{lee2012coordinated}.
Motivated by these recent trends, a transmission scheme for serving
bidirectional traffic simultaneously via spatially separated HD-BSs was investigated
in~\cite{thomsen2015compflex}. The scheme emulates Full Duplex (FD) operation using two
interconnected HD-BSs, and is termed \emph{CoMPflex}: CoMP for In-Band Wireless
Full-Duplex. In the initial work, the performance was analyzed through a simplified
one-dimensional Wyner-type deployment model. CoMPflex can be seen as a generalization of FD, where a FD BS corresponds to CoMPflex with interconnection distance zero.
We show that the nonzero separation distance in CoMPflex brings two benefits: (i) The distance between a BS and its associated Mobile Stations (MSs) decreases;
and (ii) the distance between two interfering MSs increases. This translates into
improved transmission success probability in uplink (UL) and downlink (DL).

The use of in-band FD wireless transceivers~\cite{SabharwalSGBRW13} has recently received significant
attention. However, due to the high transceiver complexity, FD is currently only feasible at the network infrastructure side~\cite{6979985} and the MSs keep the HD transceiver mode.
An in-band FD BS can serve one UL and one DL MSs simultaneously, on the same frequency. Other approaches to FD emulation by HD devices have been studied in the literature, such
as having the transmissions in UL and DL (partially) overlap in time. That is, the UL and
DL time slots, which conventionally should take place at separate time intervals, are now
overlapping, an approach used by the Rapid On-Off Division Duplexing (RODD)
in~\cite{5706936}. The authors of \cite{alammouri2015harvesting} consider the physical UL
and DL channels themselves, and have them overlap. Compared to these
approaches, CoMPflex takes advantage of the spatial dimensions in a cellular network.

In this paper, we treat CoMPflex in a two-dimensional scenario, with planar deployment of
interconnected HD-BSs. Following the trend in the literature for modeling spatial
randomness of network nodes, we analyze the performance using the tools of stochastic
geometry. Stochastic geometry has been used in many papers in the literature to model the
placement of network nodes, including FD capable ones as in~\cite{tong2015throughput}.
The setup of CoMPflex is shown in Fig.~\ref{fig:CoMPflexSystemModel}, where one HD-BS
working in the UL cooperates via a wired connection (double solid line) with another HD-BS
that operates in the DL. The solid arrows indicate signals; the dashed arrows
interference. Two interfering cells are also shown, and the boundaries between them are
indicated by dotted lines. Using the interconnection link, the interference from the
DL-BS to the UL-BS is perfectly canceled. Note that, for the sake of clarity, not all
interference from neighboring cells is shown.


\begin{figure}[t]
	\centering
		\includegraphics[width=0.8\linewidth]{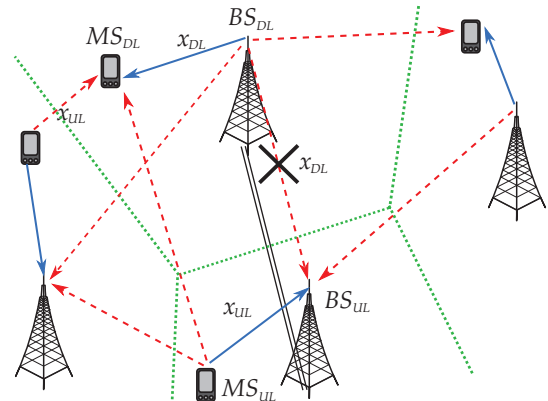}
	\caption{CoMPflex system model.}
	\label{fig:CoMPflexSystemModel}
\end{figure}

The rest of this paper is structured as follows. In Sec.~\ref{sec:SystemModel}, we
describe the system model, deployment assumptions and MS association, as well as the FD
baseline scheme. The signal model and interference characteristics are given in
Sec.~\ref{sec:SignalModel}. In Sec.~\ref{sec:analysis}, we derive and prove the
expression of the success probability for CoMPflex in UL and DL, and explain the
approximations used. The numerical results are provided and discussed in
Sec.~\ref{sec:NumericalResults}, where the effects of the CoMPflex scheme on signal and
interference link distances is shown. The paper is concluded in
Sec.~\ref{sec:Conclusion}.

\section{System Model}\label{sec:SystemModel}

We consider a scenario where HD-BSs serve HD-MSs with bidirectional traffic. We assume
Rayleigh fading with unit mean power. The power of the channel between nodes $i$ and $j$ is
written $g_{ij}$, and from the assumption of Rayleigh fading we have {$g_{ij} \sim
\Exp(1)$}.\footnote{ The notation $g \sim \Exp(\mu)$, means that $g$ is exponentially
distributed with parameter $\mu$.} The distance between nodes $i$ and $j$ is written as
$r_{ij}$. We use the pathloss model $\ell(r) = r^{-\alpha}$, where $\alpha$ is the path
loss exponent. We assume Additive White Gaussian Noise (AWGN) with power $\sigma^2$. Full
channel state information is assumed at all nodes. The default transmission power of a BS and MS is
$P_B$ and $P_M$, respectively.

\subsection{Deployment Assumptions}\label{sub:DeploymentAssumptions}

We assume that the BSs are deployed according to a Poisson Point Process (PPP)
$\Phi_{C}$ with intensity $\lambda_{C}$. The $i-$th BS located at $\xb_i \in \Rt$,
defines a \emph{Voronoi region} $\V(\xb_i)$,
\begin{equation}\label{eqn:VoronoiRegionDefinition}
	\V(\xb_i) = \left\{ \xb \in \Rt \mid \Vert \xb - \xb_i \Vert \le  \Vert \xb - \xb_j \Vert , j \neq i \right\},
\end{equation}
where $\Vert \cdot \Vert$ is the Euclidean distance. This region consists of those points
$\xb$ in $\Rt$ that are closer to the BS at $\xb_i$ than any other BS. From this definition,
the intersection of any two Voronoi regions $\V(\xb_i)$ and $\V(\xb_j)$ is empty, when $i
\neq j$. This concept will be important when we consider the MS association in
Subsec.~\ref{sub:UserAssociationandScheduling}. We assume that the Voronoi tessellation determines the rule by which the MS associates with the BS, both for DL and UL, and further that, at a specific time, only one MS randomly located in a Voronoi cell is active.

\begin{figure}[tb]
\centering
\subfloat[CoMPflex]{
\includegraphics[width=0.475\linewidth]{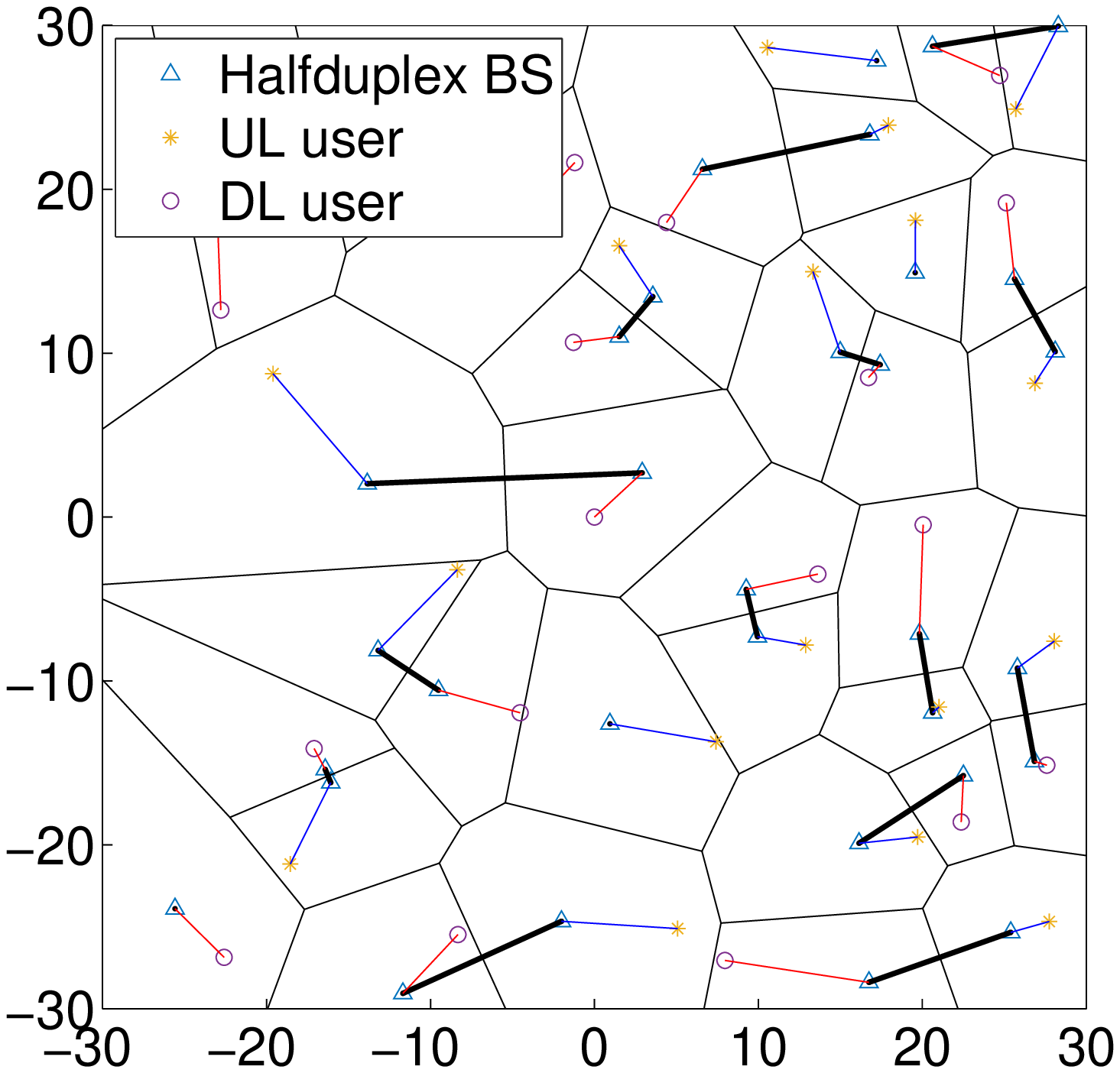}
\label{fig:CoMPflexsnapshot}}
\subfloat[Full Duplex]{
\includegraphics[width=0.475\linewidth]{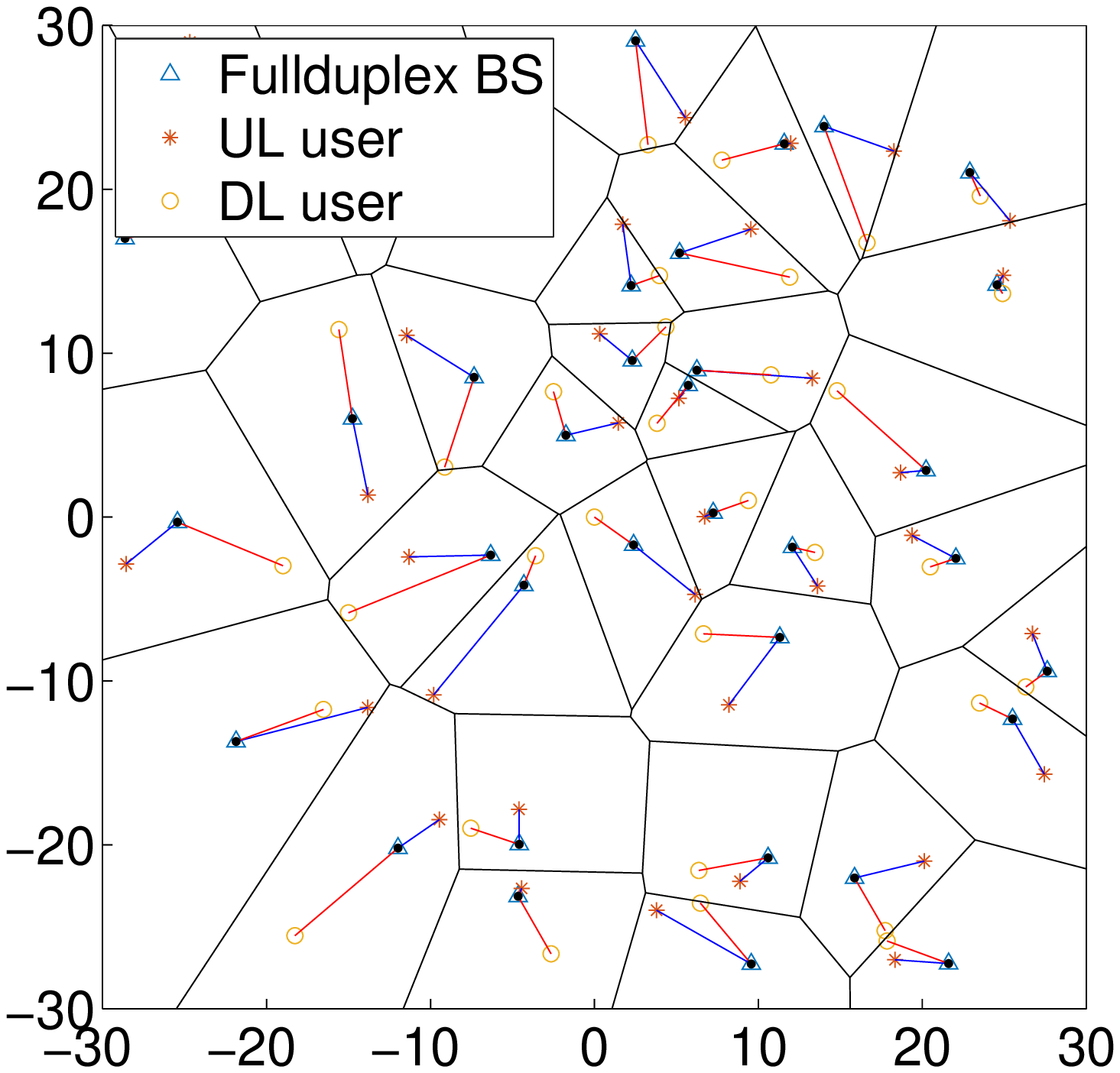}
\label{fig:FDsnapshot}}
\caption{Snapshot of network deployments. In the CoMPflex deployment, each bold line indicates a CoMPflex pair.}
\label{fig:Snapshots}
\end{figure}

\subsection{BS Pairing}\label{sub:BSPairing}

As stated previously, in CoMPflex we assume that all nodes are HD. We define a
\emph{CoMPflex pair} as two adjacent and connected HD-BSs, one serving UL and the other DL traffic. The algorithm for pairing the BSs works as follows:

Given a deployment $\Phi_{C}$, we consider a finite observation window with
dimensions $s$ km (i.e. of size $s^2$ km$^2$), and choose a BS at random in this window.
We then list all the \emph{unpaired} neighbors of this BS, and choose one of those
randomly. These two BSs are then considered to be a CoMPflex pair. The algorithm then
proceeds to the other unpaired BSs, and pairs the adjacent ones that are unpaired. The BS
closest to the origin is called the \emph{typical} BS. Fig.~\ref{fig:CoMPflexsnapshot}
shows one instance of the algorithm, where the CoMPflex pairs are indicated by bold
lines between the corresponding BSs (shown as triangles in the figure). In each CoMPflex
pair, one BS is assigned either UL or DL randomly with probability $0.5$ for each. The
other BS is assigned the opposite traffic direction. Given a MS $i$, the BS serving that
MS is denoted by $B(i)$. The algorithm terminates when it is no longer possible to pair
any BSs. By the rule of pairing, in any CoMPflex pair the two BSs must have adjacent
Voronoi regions. Therefore, any BS whose Voronoi region is surrounded by Voronoi regions
that belong to paired BSs remains unpaired. We assign the unpaired BSs either as UL or DL
at random, with probability $0.5$ for each.

After the BS pairing, all BSs have been assigned either UL or DL. Since BSs are assigned
either UL or DL at random, on average half of the BSs are UL, the other half, DL. We have $\lambda_{C} = \lambda_{C,U} + \lambda_{C,D}$, where $\lambda_{C,U}$ is the density of the CoMPflex UL-BSs, and $\lambda_{C,D}$ the density of the CoMPflex DL-BSs.

We recall that given a PPP $\Phi$ with intensity $\lambda$, we can define a new process by independently
selecting each point in $\Phi$ with probability $p$, resulting in a \emph{thinned} PPP $\Phi'$
with intensity $p\lambda$. We approximate the point processes of UL-BSs and
DL-BSs as independent thinned PPPs. This is only an approximation, as from the algorithm it follows that the selection of DL/UL is not
independent, since a DL-BS must be adjacent to a UL-BS.

\subsection{User Association and Scheduling}\label{sub:UserAssociationandScheduling}

Given the Voronoi regions defined by the BS deployment, the MS in a given region is
scheduled to be served by the BS corresponding to that region. By the definition of a
Voronoi region, a MS is associated to one unique BS. More specifically, for every Voronoi
region $\V(\xb_i)$ one MS is attached to this BS. The position of the MS is chosen
uniformly at random in the region. The traffic direction of the MS is matched to the
corresponding BS, i.e. if the BS is UL, then the MS has UL traffic, similarly for DL.

A snapshot of the deployment and pairing of BSs, along with the associated UL and DL MSs, is shown in
Fig.~\ref{fig:CoMPflexsnapshot}.

\subsection{Full Duplex Baseline Scheme}\label{sub:BaselineSchemes}

In the FD baseline scheme, the BSs are deployed according to a PPP $\Phi_{F}$ with
intensity $\lambda_{F} = 0.5\lambda_{C}$. In each Voronoi region, one UL MS and one DL MS
is served. The location of the two MSs attached to BS $i$ at location $\xb_i$ is chosen
uniformly at random inside the Voronoi region $\V(\xb_i)$ of the BS. The two MSs are
assumed to be HD devices. Note that in this setting, the number of MSs is the same as
that of CoMPflex, since each FD BS serves two MSs.

\section{Signal Model}\label{sec:SignalModel}

The UL signal to interference plus noise ratio (SINR) at BS $B(i)$ is
\begin{equation}\label{eqn:SINRUL}
	\gamma_{B(i)} = \frac{ g_{i,B(i)} \ell(r_{i,B(i)}) P_M }{ I^\psi_{B(i)} + I^\varphi_{B(i)} + \sigma^2 }.
\end{equation}
In the above, the numerator represents the UL signal. In the denominator, the term
$I^\psi_{B(i)}$ is the interference from other DL BSs, $I^\varphi_{B(i)}$ is the
interference from other UL MSs, and $\sigma^2$ is the AWGN. The interference can be
written as:
\begin{align}
	I^\psi_{B(i)} = \sum_{ u \in \psi_{B(i)}} g_{u,B(i)} \ell(r_{u,B(i)}) P_B, \\
	I^\varphi_{B(i)} = \sum_{ v \in \varphi_{B(i)}} g_{v,B(i)} \ell(r_{v,B(i)}) P_M,
\end{align}
where $\psi_{B(i)}$ and $\varphi_{B(i)}$ are the sets of interfering BSs and MSs respectively. The DL SINR at MS $j$ is given as
\begin{equation}\label{eqn:SINRDL}
	\gamma_{j} = \frac{ g_{B(j),j} \ell(r_{B(j),j}) P_B }{ I^\psi_{j} + I^\varphi_{j} + \sigma^2 }.
\end{equation}
In the above, the numerator represents the DL signal. The first term in the denominator,
$I^\psi_{j}$, is the aggregate interference from other DL BSs to DL MS $j$, and
$I^\varphi_{j}$ is the interference from the other UL MSs. The interference terms equal
\begin{align}
	I^\psi_{j} &= \sum_{ u \in \psi_{j}} g_{u,j} \ell(r_{u,j}) P_B, \\
	I^\varphi_{j} &= \sum_{ v \in \varphi_{j}} g_{v,j} \ell(r_{v,j}) P_M ,
\end{align}
where $\psi_{j}$ and $\varphi_{j}$ are the sets of interfering BSs and MSs.

\section{Reliability Analysis}\label{sec:analysis}

We analyze the performance of CoMPflex using the transmission success probability.
This metric and its complement, the outage probability, are often used in works that analyze
cellular networks through stochastic geometry. In the analysis, we consider a pair of
typical BSs and their associated MSs. The typical UL-BS is denoted $B(U)$ and the DL-BS $B(D)$,
while the UL-MS and DL-MS are denoted $U$ and $D$ respectively. These BSs and MSs represent the performance of the
entire network. We write the SINR at this BS as $\gamma_U$ (for UL). Similarly, the DL SINR
at a typical MS is written as $\gamma_D$. A transmission is successful if the SINR is not
lower that the target threshold SINR at the receiver.

\subsection{UL and DL Distance Distributions}\label{sub:ULandDLDistanceDistribution}

Recall that the BSs are deployed according to a PPP with density $\lambda_{C}$. The
distribution of the distance between a DL-BS $B(i)$ and its associated MS $i$ is denoted
$f_{r_{B(i),i}}(r)$. In deriving this distribution, we assume that the BS is located at
the origin, i.e. we consider a typical BS. The distance is then denoted
$f_{r_{U,B(U)}}(r)$. Similarly, the density of the distance between a MS $i$ and its
UL-BS $B(j)$ is written $f_{r_{B(U),U}}(r)$.

As stated from the assumptions, the location of the scheduled MS is uniform at random
inside the Voronoi region of the BS. For analytical tractability, we assume that the
location of the MS can be any point in $\Rt$. Under this assumption, the distance from the
typical BS to its MS then has the Cumulative Distribution Function (CDF):
\begin{equation}\label{eqn:ULDLdistanceCDF}
	F_{r_{B(U),U}}(r) = \Pr \{ r_{B(i),i} \le r \} = 1-\exp \left( -\lambda_{C} \pi r^2 \right),
\end{equation}
This simplification is routinely made in the literature (see
e.g.~\cite{novlan2013analytical}) for analytical tractability. We assume that the UL
distance CDF $F_{r_{U,B(U)}}(r) $  is the same as the DL. The numerical results confirm that this approximation is reasonable.

\subsection{Transmission Success Probability of CoMPflex}\label{sub:SuccessProbabilityCoMPflex}

In this section, we approximate the success probabilities in UL and DL for CoMPflex. In the
derivations, we assume that BSs and MSs are deployed according to independent PPPs with
density $\lambda_C$. Note that we approximate the locations of the MSs as a PPP, even though they are constrained to be inside the Voronoi cell of their serving BS. Also recall that the interference from the paired DL-BS to the UL-BS is cancelled, and this is reflected in the interference expressions in the proof.

\begin{thm}\label{thm:SuccessProbabilityCoMPflexUplink}
Assuming independent PPP deployment of MSs and BSs, the success probability in UL in CoMPflex is
\begin{equation}\label{eqn:SuccessProbabilityCoMPflexUL}
	P_U^{C} \!=\! 2\pi\lambda_C \!\!\int_0^\infty \!\!\!r \exp \left( -\pi\lambda_C r^2 -s\sigma^2 \right) \lt_{\psi}(s) \lt_{\varphi}(s)  \mathrm{d}r,
\end{equation}
where $s=\frac{\mu\beta_U r^\alpha}{P_M}$ and the Laplace transforms of the interference
from BSs $\lt_{\psi}(s)$ and MSs $\lt_{\varphi}(s)$ are
\begin{align}
\label{E:ltpsi}
	&\lt_{\psi}(s) = \int_0^\infty 2\pi\lambda_{C,D} t \exp \left( -\pi \lambda_{C,D} t^2 \right) \cdot \nonumber \\  & \exp \left(  -2\pi\lambda_{C,D} \int_t^\infty \frac{\beta_U \frac{P_B}{P_M} \left( \frac{r}{x} \right)^\alpha}{1+ \beta_U\frac{P_B}{P_M} \left( \frac{r}{x} \right)^\alpha} x \; \mathrm{d}x \right) \mathrm{d}t,
\end{align}
\begin{align}
\label{E:ltvarphi}
	&\lt_{\varphi}(s) \!=\! \exp \!\left(\! -2\pi\lambda_{C,U} \int_r^\infty \!\!\frac{\beta_U \left( \frac{r}{y} \right)^\alpha}{1+ \beta_U \left( \frac{r}{y} \right)^\alpha} y \; \mathrm{d}y \!\right).
\end{align}
\end{thm}
\begin{proof}
	We consider the UL SINR $\gamma_U$, and choose a typical BS. Then we condition on the distance from the BS to the nearest UL-MS being $r$. The success probability is
	\begin{align}
		P_U^C \!=\! \Pr \{ \gamma_U \!\ge\! \beta_U \} \!=\!\!\! \int_0^\infty \!\!\!\!\Pr \{ \gamma_U \ge \beta_U \!\mid\! r \} f_{U,B(U)} (r) \mathrm{d}r.  \nonumber
	\end{align}
The conditioned CDF of the SINR equals (note that we drop the explicit notation of the conditioning for readability)
\begin{align}
	&\Pr \left\{ \gamma_U \ge \beta_U \mid r \right\} = \Pr \left\{ \frac{g_{U,B(U)} r^{-\alpha} P_M }{ I_{B(U)}^\psi + I_{B(U)}^\varphi + \sigma^2 } \ge \beta_U \right\} \nonumber \\
	&\stackrel{(a)}{=} \Ebb \left[ \exp \left( - s (I_{B(U)}^\psi + I_{B(U)}^\varphi + \sigma^2)  \right) \right] \nonumber \\
	&= \exp \left( -s\sigma^2 \right) \lt_{\psi}(s) \lt_{\varphi}(s) \nonumber,
\end{align}
where in $(a)$ we have used that $g_{U,B(U)} \sim \Exp(\mu)$, and we set $s=\frac{\mu
\beta_U r^{\alpha}}{ P_M }$. We now derive the interference from the other DL-BSs. In the
derivation, we condition on the distance to the nearest interfering DL-BS to be $t$
because this distance is independent from $r$. The distance $t$ represents the
approximation of the distance to the paired DL-BS, whose transmission is perfectly
cancelled and this gives an upper bound. The density of interfering DL-BSs is
$\lambda_{C,D}$. Then
\begin{align}
	&\lt_{\psi}(s) = \Ebb_{I^\psi} \left[ \exp \left( - s \sum_{i\in \psi_{B(i)}} P_B g_{i,B(i)}r_{i,B(i)}^{-\alpha} \right) \right]  \nonumber \\
	&\stackrel{(a)}{=} \Ebb_{I^\psi} \left[ \prod_{i\in \psi_{B(i)}} \Ebb_g \left[ \exp \left( - s  P_B g_{i,B(i)}r_{i,B(i)}^{-\alpha} \right) \right] \right]  \nonumber \\
	&\stackrel{(b)}{=} \!\exp\! \left(\! -2\pi\lambda_{C,D} \!\!\!\int_t^\infty \!\!\!\left(\! 1\!-\!\Ebb_g \!\left[ \exp \!\left( \!-\! s  P_B g_{i,B(i)}x^{-\alpha} \!\right) \!\right] \!\right) x \mathrm{d}x \!\right)  \nonumber \\
	&\stackrel{(c)}{=} \exp \left( -2\pi\lambda_{C,D} \int_t^\infty \left( \frac{s P_B x^{-\alpha}}{1+s P_B x^{-\alpha}} \right) x \mathrm{d}x \right),  \nonumber
\end{align}
where in $(a)$ we have used that the channels $g_{i,B(i)}$ are independent, $(b)$ is from
the Probability Generating Functional (PGFL) of a PPP with density $\lambda_{C,D}$ and $x=r_{i,B(i)}$, and
in $(c)$ we rewrite using the Moment Generating Function (MGF) of an exponential random
variable. Combining this with the Probability Density Function (PDF) of the distance $t$ and using $s=\frac{\mu \beta_U
r^{\alpha}}{ P_M }$, we get Eq.~\eqref{E:ltpsi}. Using similar arguments, Eq.~\eqref{E:ltvarphi} also can be derived.
Note however, that in Eq.~\eqref{E:ltvarphi}, the distance to the nearest interfering
UL-MS follows the same distribution as the distance to the served UL-MS.
\end{proof}
For DL, recall that we approximate the interfering MSs as a PPP, and this approximation implies that an interfering UL-MS could be inside the Voronoi cell of the DL-MS. However, in CoMPflex there is exactly one MS in each cell. Therefore, the interference is overestimated.

\begin{thm}\label{thm:SuccessProbabilityCoMPflexDownlink}
Assuming independent PPP deployment of MSs and BSs, the success probability in DL for CoMPflex is
\begin{equation}\label{eqn:SuccessProbabilityCoMPflexDL}
	P_D^{C} \!=\! 2\pi\lambda_C \!\!\int_0^\infty \!\!\!\!r \exp \left( -\pi\lambda_C r^2 -s\sigma^2 \right) \lt_{\psi}(s) \lt_{\varphi}(s) \mathrm{d}r,
\end{equation}
where $s=\frac{\mu \beta_D r^\alpha}{P_B}$ and the Laplace transforms of the interference
from BSs $\lt_{\psi}(s)$ and MSs $\lt_{\varphi}(s)$ are
\begin{equation}
	\lt_{\psi}(s) = \exp \left( -2\pi\lambda_{C,D} \int_r^\infty \frac{ \beta_D  \left( \frac{r}{x} \right)^\alpha }{1+ \beta_D  \left( \frac{r}{x} \right)^\alpha} x \; \mathrm{d}x \right),
\end{equation}
\begin{equation}\label{eqn:LTofDLUser}
	\lt_{\varphi}(s) \!=\! \exp \left(\! -2\pi\lambda_{C,U} \!\!\int_0^\infty \!\!\frac{ \beta_D \frac{P_M}{P_B} \left( \frac{r}{y} \right)^\alpha }{1+ \beta_D \frac{P_M}{P_B} \left( \frac{r}{y} \right)^\alpha} y \; \mathrm{d}y \!\right).
\end{equation}
\end{thm}
\begin{proof}
	The proof follows similar steps as the one for UL, and so is omitted.
\end{proof}

Note that the integration range of Eq.~\eqref{eqn:LTofDLUser} starts at $0$, since there is no interference cancellation in DL, contrary to UL. From this, we obtain a lower bound on the success probability.

\subsection{Transmission Success Probability of Full Duplex}\label{sub:SuccessProbabilityFD}
In the FD baseline, for DL, Eq. (9) in~\cite{psomas2015outage} gives the outage probability of a scenario similar to our FD baseline. The success probability can be directly derived from that equation.

In deriving the UL success probability, we can use a strategy similar to the one used for DL in~\cite{psomas2015outage}. Then, the UL success probability in the FD baseline is
\begin{equation}\label{eqn:SuccessProbabilityFDUL}
	P_U^{F} \!=\! 2\pi\lambda_{F} \!\!\int_0^\infty \!\!r \exp \left( -\pi\lambda_{F}r^2 \!-\!s\sigma^2 \right) \lt_{\psi}(s) \lt_{\varphi}(s) \; \mathrm{d}r,
\end{equation}
where $s=\frac{\mu\beta_U r^\alpha}{P_M}$ and the Laplace transforms of the interference
from BSs $\lt_{\psi}(s)$ and MSs $\lt_{\varphi}(s)$ are
\begin{equation}
	\lt_{\psi}(s) = \exp \left( -2\pi\lambda_F \int_r^\infty \frac{\beta_U \frac{P_B}{P_M} \left( \frac{r}{x} \right)^\alpha }{1+\beta_U \frac{P_B}{P_M} \left( \frac{r}{x} \right)^\alpha} x \; \mathrm{d}x \right),
\end{equation}
\begin{equation}
	\lt_{\varphi}(s) = \exp \left( -2\pi\lambda_F \int_r^\infty \frac{\beta_U \left( \frac{r}{y} \right)^\alpha }{1+\beta_U \left( \frac{r}{y} \right)^\alpha} y \; \mathrm{d}y \right).
\end{equation}

\section{Numerical Results}\label{sec:NumericalResults}

We show the performance of CoMPflex, and the comparison with the FD baseline schemes, using both numerical simulations and the analytical model given in the previous section. The simulation assumptions are shown in Table~\ref{tab:SimulationParameters}, where the densities are chosen comparable with~\cite{psomas2015outage}.

\begin{table}[h]
	\caption{Simulation parameters.}
	\centering
			\begin{tabular}{ c l c }
		  	\hline
  			Parameter & Description & Simulation Setting\\
				\hline
				$s$ & Size of observation window & $200$ km \\
				$\lambda_{C}$ & BS density (CoMPflex) & $0.02$ ${\text{BS}}/{\text{km}^2}$ \\
				$\lambda_{F}$ & BS density (FD) & $0.01$ ${\text{BS}}/{\text{km}^2}$ \\
				$\sigma^2$ & Noise power at MS and BS & $-174$ dBm \\
				$\alpha$ & Path loss exponent & $4$ \\
				$\beta$ & SINR thresholds & $-20, -15, -10,\ldots, 20$ dB \\
				$P_B$ & BS transmission power & $40$ dBm \\
				$P_M$ & MS transmission power & $20$ dBm \\
				\hline
			\end{tabular}
	\label{tab:SimulationParameters}
\end{table}

We study the success probability in both UL and DL for CoMPflex, in terms of varying the
SINR threshold, and compare with FD. The BS and MS transmission powers are held constant
according to the values in Tab.~\ref{tab:SimulationParameters}. The UL success
probability for CoMPflex and FD, both simulation and analytical, is shown in
Fig.~\ref{fig:SuccessProbabilityULComparison}. Here, the analytical curve approximates
the simulated values quite closely. However, the success probability in UL is lower than
DL, which can be partially explained by the MS power being lower than the BS power.

The resulting success probability for DL is shown in Fig.~\ref{fig:SuccessProbabilityDLComparison}. In this figure, we can observe that the analytical derivations result in a lower bound on the success probability. This was to be expected, since the point processes in CoMPflex are not truly PPP. However, as the figure shows, the PPP approximation is quite close and serves well as an indicator of the expected performance of CoMPflex. We also observe that the success probability in CoMPflex is about $30 \%$ higher than FD, for most of the range of SINR thresholds. The explanation of this can be attributed to how CoMPflex affects the distances of the signal and interference links.

\begin{figure}[t]
	\centering
		\includegraphics[width=0.85\linewidth]{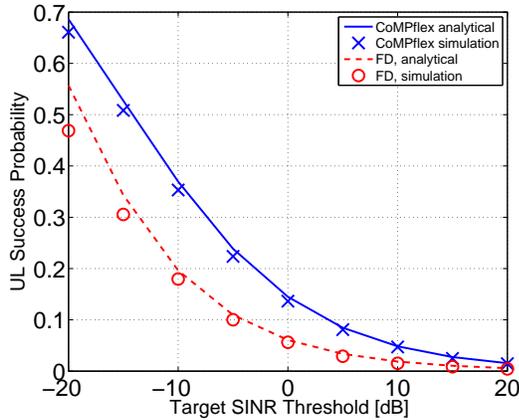}
	\caption{Success Probability in UL vs. SINR threshold.}
	\label{fig:SuccessProbabilityULComparison}
\end{figure}

\begin{figure}[t]
	\centering
		\includegraphics[width=0.85\linewidth]{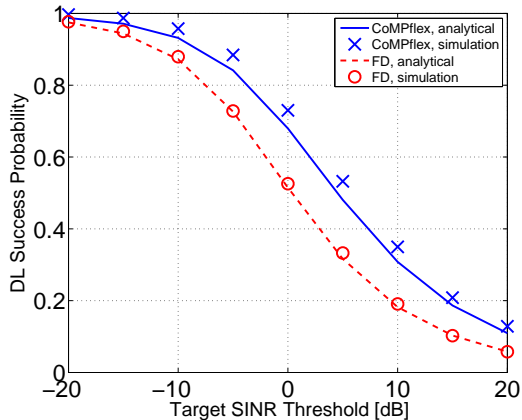}
	\caption{Success Probability in DL vs. SINR threshold.}
	\label{fig:SuccessProbabilityDLComparison}
\end{figure}

Since one of the main features in CoMPflex is that it brings the MSs closer to the serving BSs, we analyze and compare the CDFs of the various distances for signal and intra-cell interference links in CoMPflex and FD. This comparison is shown in Fig.~\ref{fig:CDFdistances}. The simulated CDF are shown as lines, while the analytical CDF using Eq.~\eqref{eqn:ULDLdistanceCDF} are shown as markers\footnote{Note that the analytical CDF of CoMPflex is shifted to the left, compared to FD, since $\lambda_{F} = 0.5 \lambda_{C}$}. From this figure, we can observe two important points:

First, we compare the CDFs of the distances between an UL-MS and BS, and between a BS and DL-MS. We see that for both CoMPflex and FD, the UL and DL distance curves overlap. This implies that the distances of UL and DL follow the same distribution. What is also interesting is that the CDFs of the distances in CoMPflex are shifted to the \emph{left}, compared to FD. This means that the lower MS to BS distances have higher probability in CoMPflex compared to FD.

Second, the CDF curve of the intra-cell interference distance in CoMPflex is shifted to the \emph{right} compared to FD. This means that higher interference distances have higher probability in CoMPflex compared to FD. Taken together, these two points can explain the performance advantages of CoMPflex over FD, which come from having a lower signal distance and a higher interference distance \emph{simultaneously}.

\begin{figure}[t]
	\centering
		\includegraphics[width=0.85\linewidth]{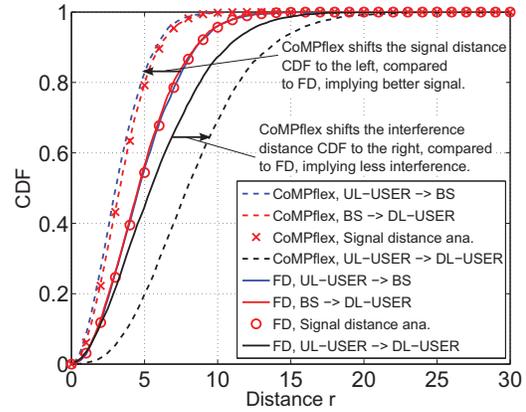}
	\caption{Comparison of the CDFs of the node distances in the CoMPflex and FD scenarios.}
	\label{fig:CDFdistances}
\end{figure}

\section{Conclusion}\label{sec:Conclusion}

In this work, we have analyzed the performance of CoMPflex in a planar network setting, and compared it with a FD baseline scheme. We have derived the success probability for UL and DL, and validated the results via simulations. It was observed that the success probability of both UL and DL was higher for CoMPflex than in FD, an effect which can be attributed to the effects of CoMPflex on the node distances. Thus it is beneficial to consider the usage of HD devices instead of FD, which is also useful for using already existing technology and avoiding the signal complexities of FD.

One promising research direction is to consider more than two connected BSs, and more general clustering criteria. Also, it would be interesting to compare CoMPflex with other interference mitigation techniques such as CoMP.

\section*{Acknowledgment}

This work was supported by Innovation Fund Denmark, via the Virtuoso project.

\end{document}